\documentclass[a4paper,10pt]{amsart}

\usepackage{amssymb,amscd,amsmath}

\theoremstyle{plain}

\newtheorem{thm}{Theorem}[section]

\newtheorem{lem}[thm]{Lemma}
\newtheorem{cor}[thm]{Corollary}
\newtheorem{obs}[thm]{Observation}

\theoremstyle{definition}

\theoremstyle{remark}

\newcommand{\forme}[1]{}

\title{The Italian bondage and reinforcement numbers of digraphs}

\author[Kim]{Kijung Kim}
\address{Department of Mathematics, Pusan National University, Busan 46241, Republic of Korea}
\email{knukkj@pusan.ac.kr}

\date{\today}
\begin{document}

\begin{abstract}
An \textit{Italian dominating function} on a digraph $D$ with vertex set $V(D)$
is defined as a function $f : V(D) \rightarrow \{0, 1, 2\}$  such that every vertex $v \in V(D)$ with
$f(v) = 0$ has at least two in-neighbors assigned $1$ under $f$ or one in-neighbor $w$ with $f(w) = 2$.
The \textit{weight} of an Italian dominating function $f$ is the value $\omega(f) = f(V(D)) = \sum_{u \in V(D)} f(u)$.
The \textit{Italian domination number} of a digraph $D$, denoted by $\gamma_I(D)$, is the minimum taken over the weights of all Italian dominating functions on $D$.
The \textit{Italian bondage number} of a digraph $D$, denoted by $b_I(D)$, is the minimum number of arcs of $A(D)$ whose removal in $D$
results in a digraph $D'$ with $\gamma_I(D') > \gamma_I(D)$.
The \textit{Italian reinforcement number} of a digraph $D$, denoted by $r_I(D)$, is the minimum number of extra arcs whose addition to $D$
results in a digraph $D'$ with $\gamma_I(D') < \gamma_I(D)$.
In this paper, we initiate the study of Italian bondage and reinforcement numbers in digraphs and present
some bounds for $b_I(D)$ and $r_I(D)$.
We also determine the Italian bondage and reinforcement numbers of some classes of digraphs.

\bigskip

\end{abstract}

\maketitle

\insert\footins{\footnotesize
This research was supported by Basic Science Research Program through the National Research Foundation of Korea funded by the Ministry of Education (2020R1I1A1A01055403).}

\section{Introduction}\label{sec:intro}

Let $D=(V,A)$ be a finite simple digraph with vertex set $V=V(D)$ and arc set $A=A(D)$.
The order $n(D)$ of a digraph is the size of $V(D)$.
For an arc $uv \in A(D)$, we say that $v$ is an \textit{out-neighbor} of $u$ and
$u$ is an \textit{in-neighbor} of $v$.
We denote the set of in-neighbors and out-neighbors of $v$ by $N_D^-(v)$ and $N_D^+(v)$, respectively.
We write $deg_D^-(v)$ and $deg_D^+(v)$ for the size of $N_D^-(v)$ and $N_D^+(v)$, respectively.
Let $N_D^-[v] = N_D^-(v) \cup \{ v\}$ and $N_D^+[v] = N_D^+(v) \cup \{ v\}$.
For s subset $S$ of $V(D)$, we define $N^+(S) = \bigcup_{v \in S} N_D^+(v)$ and $N^+[S] = \bigcup_{v \in S} N_D^+[v]$.
The \textit{maximum out-degree} and \textit{maximum in-degree} of a digraph $D$ are denoted by $\Delta^+(D)$ and $\Delta^-(D)$, respectively.

For a digraph $D$, a subset $S$ of $V(D)$ is a \textit{dominating set} if $\bigcup_{v \in S} N_D^+[v] = V(D)$.
The \textit{domination number} $\gamma(D)$ is the minimum cardinality of a dominating set of $D$.
The concept of the domination number of a digraph was introduced in \cite{CHY}.
The \textit{bondage number} $b(D)$ of a digraph $D$ is the minimum number of arcs of $A(D)$ whose removal in $D$
results in a digraph $D'$ with $\gamma(D') > \gamma(D)$.
The concept of the bondage number of a digraph was proposed in \cite{CD}.
The \textit{reinforcement number} $r(D)$ of a digraph $D$ is the minimum number of extra arcs whose addition to $D$
results in a digraph $D'$ with $\gamma(D') < \gamma(D)$.
The concept of the reinforcement number of a digraph was introduced in \cite{HWX}.

Among the variations of domination, so called Italian domination of graphs is introduced in \cite{CHHM}.
The authors of \cite{CHHM} present bounds relating the Italian domination number to some other domination parameters.
The authors of \cite{HK} characterize the trees $T$ for which $\gamma(T) +1 =\gamma_I(T)$ and also
characterize the trees $T$ for which $\gamma_I(T) =2\gamma(T)$.
After that, there are many studies on Italian domination of graphs in \cite{GWLY, GXY, LSX, SSSM, RC}.
Recently, the author of \cite{Vol} initiated the study of the Italian domination number in digraphs.
Related results was given in \cite{Kim, Vol-2}.
Our aim in this paper is to initiate the study of Italian bondage and reinforcement numbers for digraphs.

An \textit{Italian dominating function} (IDF) on a digraph $D$ with vertex set $V(D)$
is defined as a function $f : V(D) \rightarrow \{0, 1, 2\}$  such that every vertex $v \in V(D)$ with
$f(v) = 0$ has at least two in-neighbors assigned $1$ under $f$ or one in-neighbor $w$ with $f(w) = 2$.
An Italian dominating function $f : V(D) \rightarrow \{0, 1, 2\}$ gives an ordered partition $(V_0, V_1, V_2)$ (or $(V_0^f, V_1^f, V_2^f)$ to refer to $f$) of $V(D)$, where $V_i:= \{ x \in V(D) \mid f(x)=i \}$.
The \textit{weight} of an Italian dominating function $f$ is the value $\omega(f) = f(V(D)) = \sum_{u \in V(D)} f(u)$.
The \textit{Italian domination number} of a digraph $D$, denoted by $\gamma_I(D)$, is the minimum taken over the weights of all Italian dominating functions on $D$.
A \textit{$\gamma_I(D)$-function} is an Italian dominating function on $D$ with weight $\gamma_I(D)$.

The \textit{Italian bondage number} of a digraph $D$, denoted by $b_I(D)$, is the minimum number of arcs of $A(D)$ whose removal in $D$
results in a digraph $D'$ with $\gamma_I(D') > \gamma_I(D)$.

The \textit{Italian reinforcement number} of a digraph $D$, denoted by $r_I(D)$, is the minimum number of extra arcs whose addition to $D$
results in a digraph $D'$ with $\gamma_I(D') < \gamma_I(D)$.
The Italian reinforcement number of a digraph $D$ is defined to be $0$ if $\gamma_I(D) \leq 2$.
A subset $R$ of $A(\overline{D})$ is called an  \textit{Italian reinforcement set} (IRS) of $D$ if $\gamma_I(D +R) < \gamma_I(D)$.
An $r_I(D)$-set is an IRS of $D$ with size $r_I(D)$.

This paper is organized as follows.
In Section \ref{sec:main1}, we prepare basic results on the Italian domination number.
In Section \ref{sec:main2}, we give some bounds of the Italian bondage number and determine the exact values of Italian bondage numbers of some classes of digraphs.
In Section \ref{sec:main3}, we characterize all digraphs $D$ with $r_I(D)=1$.
We give some bounds of the Italian reinforcement number and also determine the exact values of Italian reinforcement numbers of compositions of digraphs.

\section{The Italian domination numbers}\label{sec:main1}

In this paper, we make use of the following results.

\begin{obs}\label{easy1}
For a digraph $D$, $\gamma_I(D) \leq n - \Delta^+(D) + 1$.
\end{obs}
\begin{proof}
Let $D$ be a digraph, and let $v$ be a vertex with $deg_D^+(v) =\Delta^+(D)$.
Define a function $f : V(D) \rightarrow \{0,1,2\}$ by $f(v)=2$, $f(x)= 0$ if $x \in N^+(v)$, and $f(x)=1$ otherwise.
It is easy to see that $f$ is an IDF of $D$.
\end{proof}

The following result is the exact value of Italian domination number of a complete bipartite graph
(see \cite{CHHM} for the definition of Italian dominating function and domination number on a graph).

\begin{lem}[\cite{HSW}]\label{domination}
For a complete bipartite graph $K_{m,n}$ with $1 \leq m \leq n$ and $n \geq 2$,
\begin{equation*}\label{product}
\gamma_I(K_{m,n}) = \left\{
                      \begin{array}{ll}
                     2 & \hbox{if $m \leq 2$;} \\
                     3 & \hbox{if $m = 3$;} \\
                     4 & \hbox{if m $\geq 4$.}
                      \end{array}
                     \right.
\end{equation*}
\end{lem}





\begin{thm}[\cite{Vol}]\label{mainthm2}
Let $D$ be a digraph of order $n$. Then $\gamma_I(D) \geq \lceil \frac{2n}{2 + \Delta^+(D)} \rceil$.
\end{thm}

\begin{thm}\label{mainthm2-5}
Let $D$ be a digraph of order $n \geq 3$. Then $\gamma_I(D) =2$ if and only if $\Delta^+(D)=n-1$ or
there exist two distinct vertices $u$ and $v$ such that $V(D) \setminus \{u, v\} \subseteq N_D^+(u)$ and $V(D) \setminus \{u, v\} \subseteq N_D^+(v)$.
\end{thm}
\begin{proof}
If $\Delta^+(D)=n-1$ or
there exist two distinct vertices $u$ and $v$ such that $V(D) \setminus \{u, v\} \subseteq N_D^+(u)$ and $V(D) \setminus \{u, v\} \subseteq N_D^+(v)$, then it is easy to see that $\gamma_I(D) =2$.

Assume that $\gamma_I(D) =2$.
Let $(V_0, V_1, V_2)$ be a $\gamma_I(D)$-function.
Then $\gamma_I(D) = 2 = |V_1| + 2|V_2|$ and $|V_2| \leq 1$.
If $|V_2|=1$, then $|V_1|=0$ and hence $\Delta^+(D)=n-1$.
If $|V_2|=0$, then $|V_1|=2$ and, by the definition of IDF,
there exist two distinct vertices $u$ and $v$ such that $V(D) \setminus \{u, v\} \subseteq N_D^+(u)$ and $V(D) \setminus \{u, v\} \subseteq N_D^+(v)$
\end{proof}

\begin{thm}[\cite{Vol}]\label{mainthm3}
Let $D$ be a digraph of order $n \geq 3$. Then $\gamma_I(D) < n$ if and only if $\Delta^+(D) \geq 2$ or $\Delta^-(D) \geq 2$.
\end{thm}

\begin{cor}\label{mainthm3-cor}
If $D$ is a directed path or cycle of order $n$, then $\gamma_I(D) = n$.
\end{cor}

\section{The Italian bondage numbers}\label{sec:main2}

\subsection{Bounds of the Italian bondage numbers}\label{sec:main2-1}

The \textit{underlying graph} $G[D]$ of a digraph $D$ is the graph obtained by replacing each arc $uv$ by an edge $uv$.
Note that $G[D]$ has two parallel edges $uv$ when $D$ contains the arc $uv$ and $vu$.
A digraph $D$ is \textit{connected} if the underlying graph $G[D]$ is connected.
For a graph $G$, we denote the degree of $v \in V(G)$ by $deg_G(v)$.
In particular, $\Delta(G)$ means the maximum degree in $G$.

\begin{thm}\label{mainthm9}
If $D$ is a digraph, and $xyz$ a path of length $2$ in $G[D]$ such that $yx, yz \in A(D)$,
then
\[b_I(D) \leq deg_{G[D]}(x) + deg_{D}^-(y) + deg_{G[D]}(z) - |N^-(x) \cap N^-(y) \cap N^-(z)|.\]
Moreover, if $x$ and $z$ are adjacent in $G[D]$, then
\[b_I(D) \leq deg_{G[D]}(x) + deg_{D}^-(y) + deg_{G[D]}(z) -1 - |N^-(x) \cap N^-(y) \cap N^-(z)|.\]
\end{thm}
\begin{proof}
Let $B$ be the set of all arcs incident with $x$ or $z$ and all arcs terminating at $y$ with the exception of all arcs
from $N^-(x) \cap N^-(z)$ to $y$. Then
\[ |B| \leq deg_{G[D]}(x) + deg_{D}^-(y) + deg_{G[D]}(z) - |N^-(x) \cap N^-(y) \cap N^-(z)| \]
and
\[ |B| \leq deg_{G[D]}(x) + deg_{D}^-(y) + deg_{G[D]}(z) -1 - |N^-(x) \cap N^-(y) \cap N^-(z)| \]
when $x$ and $z$ are adjacent.

Let $D' = D - B$.
In $D'$, $x$ and $z$ are isolated, and all in-neighbors of $y$ in $D'$, if any, lie in $N^-(x) \cap N^-(z)$.
Let $f=(V_0, V_1, V_2)$ be a $\gamma_I(D')$-function.
Then $f(x)=f(z)=1$.
If $f(y)=2$, then
\[(V_0 \cup \{x,z\}, V_1 \setminus \{x,z\}, V_2)\]
 is an IDF of $D$ with weight less than $\omega(f)$.
If $f(y)=1$, then
\[(V_0 \cup \{x,z\}, V_1 \setminus \{x,y,z\}, V_2 \cup \{y\})\]
is an IDF of $D$ with weight less than $\omega(f)$.
However, if $f(y)=0$, then there exists $w \in N^-(x) \cap N^-(y) \cap N^-(z)$ such that $f(w)=2$ or
there exist $w_1, w_2 \in N^-(x) \cap N^-(y) \cap N^-(z)$ such that $f(w_1)=f(w_2)=1$.
Since $w, w_1$ and $w_2$ are in-neighbors of $x$ and $z$ in $D$,
\[(V_0 \cup \{x,z\}, V_1 \setminus \{x,z\}, V_2)\]
is an IDF of $D$ with weight less than $\omega(f)$.
This completes the proof.
\end{proof}

\begin{thm}\label{mainthm10}
Let $D$ be a digraph of order $n \geq 3$.
If $G[D]$ is connected, then
\[b_I(D) \leq (\gamma_I(D)-1)\Delta(G[D]).\]
\end{thm}
\begin{proof}
We proceed by induction on $\gamma_I(D)$.
Assume that $\gamma_I(D)=2$.
For a vertex $u \in V_1 \cup V_2$, let $B_u$ be the set of arcs incident with $u$.
Since $\gamma_I(D - u) \geq 2$ by $n \geq 3$,
we have \[\gamma_I(D -B_u) = \gamma_I(D - u) +1 \geq 3.\]
This implies that $b_I(D) \leq |B_u|$ for $u \in V_1 \cup V_2$. Thus, $b_I(D) \leq \Delta(G[D])$.

Assume that the result is true for every digraph with the Italian domination number $k \geq 3$.
Let $D$ be a digraph with $\gamma_I(D)=k+1$.
Suppose to the contrary that $b_I(D) > (\gamma_I(D)-1)\Delta(G[D])$.
Let $u$ be an arbitrary vertex of $D$, and let $B_u$ be the set of arcs incident with $u$.
Then we have $\gamma_I(D) = \gamma_I(D - B_u)$.
Let $f$ be a $\gamma_I(D - B_u)$-function.
Then $f(u)=1$ and the function $f$ restricted to $D - u$ is also a $\gamma_I(D - u)$-function.
This implies that $\gamma_I(D - u) = \gamma_I(D) -1$.
So, $b_I(D) \leq b_I(D -u) + deg_{G[D]}(u)$.
By the induction hypothesis, we have
\begin{eqnarray*}
b_I(D)  & \leq &  b_I(D -u) + deg_{G[D]}(u) \\
              & \leq & (\gamma_I(D - u)-1)\Delta(G[D -u]) + deg_{G[D]}(u)  \\
              & \leq & (\gamma_I(D - u)-1)\Delta(G[D]) + \Delta(G[D]) \\
              & = & \gamma_I(D - u)\Delta(G[D]) \\
              & = & (\gamma_I(D) -1)\Delta(G[D]).
\end{eqnarray*}
This is a contradiction.
\end{proof}

\subsection{The Italian bondage numbers of some classes of digraphs}\label{sec:main2-2}

For a graph $G$, the \textit{associated digraph} $G^*$ is the digraph obtained from $G$ by replacing each edge of $G$ by two oppositely oriented arcs. Note that $\gamma_I(G) = \gamma_I(G^*)$ for any graph $G$.

\begin{thm}\label{mainthm11}
Let $K_n^*$ be the complete digraph of order $n \geq 3$.
Then $b_I(K_n^*)= n$.
\end{thm}
\begin{proof}
Note that $\gamma_I(K_n^*)=2$.
Let $B$ be an arc set of $K_n^*$. Define $D:= K_n^* - B$.
If $D$ contain a vertex $x$ such that $deg_D^+(x) = n-1$, then it follows from Observation \ref{easy1} that $\gamma_I(D)=2$.
This implies that $b_I(K_n^*) \geq n$.

Let $\{x_1,x_1, \dotsc, x_n\}$ be the vertex set of $K_n^*$, and let $B:=\{x_1x_2, x_2x_3, \dotsc, x_nx_1\}$ be the arc set of a directed
cycle in $K_n^*$.
Define $D:= K_n^* - B$. Then one can observe that there do not exist two distinct vertices $u$ and $v$ in $D$ such that
$V(D) \setminus \{u, v\} \subseteq N_D^+(u)$ and $V(D) \setminus \{u, v\} \subseteq N_D^+(v)$.
It follows from Theorem \ref{mainthm2-5} that $\gamma_I(D) \geq 3$.
This completes the proof.
\end{proof}

The following result follows from the definition of associated digraph and Lemma \ref{domination}.
For a complete bipartite digraph $K_{m,n}^*$ with $1 \leq m \leq n$,
\begin{equation}\label{com-bi}
\gamma_I(K_{m,n}^*) = \left\{
                      \begin{array}{ll}
                     2 & \hbox{if $m \leq 2$;} \\
                     3 & \hbox{if $m = 3$;} \\
                     4 & \hbox{if m $\geq 4$.}
                      \end{array}
                     \right.
\end{equation}

\begin{thm}\label{mainthm12}
Let $K_{m,n}^*$ be the complete bipartite digraph such that $1 \leq m < n$.
Then
\begin{equation*}\label{com-mul}
b_I(K_{m,n}^*) = \left\{
                      \begin{array}{ll}
                     1 & \hbox{if $m \leq 2$;} \\
                     2 & \hbox{if $m=3$;} \\
                     m+2 & \hbox{if $m \geq 4$.}
                      \end{array}
                     \right.
\end{equation*}
\end{thm}
\begin{proof}
We denote $K_{m,n}^*$ by $D$.
Let $X =\{x_1, x_2, \dotsc, x_m\}$ and $Y= \{y_1, y_2, \dotsc, y_n\}$ be the partite sets of $D$.
The result is clear for $m \leq 2$.

Assume that $m=3$. It follows from (\ref{com-bi}) that $\gamma_I(D)=3$.
If we remove two arcs terminating at some vertex $y_j \in Y$,
then the Italian domination number of resulting digraph increases. So, $b_I(D) \leq 2$.
For any arc $e$ of $A(D)$, there exist two vertices $x_i$ and $x_j$ such that $N_{D-e}^+(x_i)=n$ and $N_{D-e}^+(x_j)=n$.
Thus, we have $b_I(D)=2$.

Assume that $m \geq 4$. It follows from (\ref{com-bi}) that $\gamma_I(D)=4$.
Let $B= \{ x_iy_1 \mid 1 \leq i \leq m \} \cup \{y_1x_1, y_1x_2 \}$.
It is easy to see that $\gamma_I(D -B) \geq 5$. So, $b_I(D) \leq m+2$.

Next, we show that $b_I(D) \geq m+2$.
Let $B'$ be a subset of $A(D)$ such that $|B'| = m+1$, and let $D' = D - B'$.
Then $D'$ has at least $n-1$ vertices whose outdegree are equal in $D$ and $D'$.
Let $E = \{ v \in V(D) \mid d_D^+(v) = d_{D'}^+(v) \}$.
If $E \cap X \neq \emptyset \neq E \cap Y$, then clearly $\gamma_I(D')=4$.
Henceforth, we assume that $E \cap X = \emptyset$ or $E \cap Y = \emptyset$.
Without loss of generality, assume that $E \cap X = \emptyset$.
Then $E \subseteq Y$ and $B'$ contains one outgoing arc for each $x_i \in X$.
Since $|B'|=m+1 < 2m$, $B'$ contains exactly one outgoing arc for some $x_i \in X$.
Without loss of generality, assume that $i=1$ and $x_1y_1 \in B'$.
If $E=Y$, then
\[(V(D') \setminus \{x_1,y_1\}, \emptyset, \{x_1,y_1\})\]
is an IDF of $D'$ with weight $4$.
Let $E \subset Y$. We may assume that $E \subseteq \{y_1, y_2, \dotsc, y_{n-1}\}$.
Thus, $B'$ contains one outgoing arc from $y_n$, say $y_nx_m$. Since $|B'|=m+1$, $B'$ contains exactly one outgoing arc for each $x_i \in X$
and one outgoing arc from $y_n$.
If $x_iy_j \in B'$ for some $1 \leq i \leq m$ and $j < n$, then
\[(V(D') \setminus \{x_i, y_j\}, \emptyset, \{x_i, y_j\})\]
is an IDF of $D'$ with weight $4$. Thus, we assume that $x_iy_n \in B'$ for each $1 \leq i \leq m$.
But,
\[(V(D') \setminus \{x_m, y_n\}, \emptyset, \{x_m, y_n\})\]
is an IDF of $D'$ with weight $4$.
Thus, we have $b_I(D) \geq m+2$.
\end{proof}

\section{The Italian reinforcement numbers}\label{sec:main3}

\subsection{Digraphs with $r_I(D)=1$}\label{sec:main3-1}

\begin{lem}\label{lem:b1}
Let $D$ be a digraph with $\gamma_I(D) \geq 3$.
Let $F$ be an $r_I(D)$-set, and let $g$ be a $\gamma_I(D)$-function of $D + F$.
Then the following hold:
\begin{enumerate}
\item For each arc $v_1v_2 \in F$, $g(v_1) \neq 0$ and $g(v_2) = 0$ .
\item $\gamma_I(D+F) = \gamma_I(D)- 1$.
\end{enumerate}
\end{lem}
\begin{proof}
If there exists an arc $v_1v_2 \in F$ such that either $g(v_i)\geq 1$ for each $i \in \{1,2\}$ or $g(v_1)=g(v_2)=0$, then
$g$ is also an IDF of $D +(F \setminus \{v_1v_2\})$, and hence $F \setminus \{v_1v_2\}$ is an IRS of $D$, which contradicts the definition of $F$.
Thus, (i) holds.

By the definition of $F$, we have $\gamma_I(D +F) \leq \gamma_I(D)- 1$.
Suppose that $\gamma_I(D +F) \leq \gamma_I(D)- 2$.
Let $v_1v_2 \in F$. By (i), $g(v_1)\neq 0$ and $g(v_2)=0$.
Then the function $g' : V(D + (F \setminus \{v_1v_2\})) \rightarrow \{0,1,2\}$ with
\begin{equation*}
g'(x) = \left\{
                      \begin{array}{ll}
                     1 & \hbox{if $x=v_2$;} \\
                     g(x) & \hbox{otherwise}
                      \end{array}
                     \right.
\end{equation*}
is an IDF of $D + (F \setminus \{v_1v_2\})$ such that $\omega(g')=\omega(g)+1 \leq \gamma_I(D)- 1$.
This implies that $F \setminus \{v_1v_2\}$ is an IRS of $D$, which contradicts the definition of $F$.
Thus, (ii) holds.
\end{proof}

\begin{lem}\label{irs=1}
Let $D$ be a digraph of order $n \geq 3$, $\Delta^+(D) \geq 1$ and $\gamma_I(D) = n$.
Then $r_I(D) = 1$.
\end{lem}
\begin{proof}
It follows from Theorem \ref{mainthm3} that $\Delta^+(D) =1$.
Since $\sum_{v \in V(D)} deg^+(v) =  \sum_{v \in V(D)} deg^-(v)$, we have $\Delta^-(D) \geq 1$.
It also follows from Theorem \ref{mainthm3} that $\Delta^-(D) =1$.
Thus, $D$ is disjoint union of directed paths, cycles or isolated vertices.
Let $uv \in A(D)$ and $w \in V(D) \setminus \{u, v\}$.
It is easy to see that
\[(\{v, w\}, V(D) \setminus \{u,v,w\}, \{u\})\]
is an IDF of $D + uw$ with weight $n-1$.
Thus, we have $r_I(D) = 1$.
\end{proof}

\begin{thm}\label{m-thm1}
Let $D$ be a digraph with $\gamma_I(D) \geq 3$. Then $r_I(D) = 1$ if and only if
there exist a $\gamma_I(D)$-function $f=(V_0, V_1, V_2)$ of $D$ and a vertex $v \in V_1$ satisfying one of the following conditions:
\begin{enumerate}
\item $f(N^-(v))=1$ and $f(N^-(x) \setminus \{v\}) \geq 2$ for each $x \in N^+(v) \cap V_0$.
\item $f(N^-(v))=0$, $f(N^-(x) \setminus \{v\}) \geq 2$ for each $x \in N^+(v)$, and $V_2 \neq \emptyset$.
\end{enumerate}
\end{thm}
\begin{proof}
First, assume that (i) holds.
Then it follows from $f(N^-(v))=1$ that there exists $u \in V_1 \cap N^-(v)$.
Since $\gamma_I(D) \geq 3$, there exists $w \in (V_1 \cup V_2) \setminus \{v, u\}$.
Since $uv \in A(D)$ and $f(N^-(x) \setminus \{v\}) \geq 2$ for each $x \in N^+(v) \cap V_0$,
$(V_0 \cup \{v\}, V_1 \setminus \{v\}, V_2)$ is an IDF of $D +wv$ with weight $\gamma_I(D) -1$.
Thus, we have $r_I(D) = 1$.

Next, assume that (ii) holds.
Let $w \in V_2$.
Then it follows from $f(N^-(v))=0$ that $wv \not\in A(D)$.
Since $f(N^-(x) \setminus \{v\}) \geq 2$ for each $x \in N^+(v)$,
$(V_0 \cup \{v\}, V_1 \setminus \{v\}, V_2)$ is an IDF of $D +wv$ with weight $\gamma_I(D) -1$.
Thus, we have $r_I(D) = 1$.

Conversely, assume that $r_I(D) = 1$, and let $uv$ be an arc of $D$ with $\gamma_I(D+uv) <  \gamma_I(D)$.
Let $g$ be a $\gamma_I(D+uv)$-function.
Then $g(u) \neq 0$ and $g(v)=0$ by Lemma \ref{lem:b1}(i).
The function $f : V(D) \rightarrow \{0,1,2\}$ with
\begin{equation*}\label{product}
f(x) = \left\{
                      \begin{array}{ll}
                     1 & \hbox{if $x=v$;} \\
                     g(x) & \hbox{otherwise}
                      \end{array}
                     \right.
\end{equation*}
is an IDF of $D$. It follows from Lemma \ref{lem:b1}(ii) that $f$ is a $\gamma_I(D)$-function.

Suppose that $f(N^-(v)) \geq 2$.
Then $g(N^-(v)) \geq 2$. So, $g$ is an IDF of $D$. This means that $\gamma_I(D) \leq \omega (g) = \gamma_I(D+uv)$, a contradiction.
Thus, we have $f(N^-(v)) \leq 1$.

Note that $f(N^-(x) \setminus \{v\}) = h(N^-(x) \setminus \{v\})  \geq 2$ for each $x \in N^+(v) \cap V_0^f$,
since $g$ is a $\gamma_I(D+uv)$-function with $g(v)=0$.
If $f(N^-(v)) = 1$, then (i) holds.
Now assume that $f(N^-(v)) = 0$.
Then we have $h(u) = f(u) =2$, since $g(v)= 0$ and $u$ is an in-neighbor of $v$ in $D + uv$.
As $V_2^f \neq \emptyset$, (ii) holds.
\end{proof}

\subsection{Bounds of the Italian reinforcement numbers}\label{sec:main3-2}

\begin{thm}
If D is a digraph of order $n$ with  $\gamma_I(D) \geq 3$, then
\[r_I(D) \leq n - \Delta^+(D) - \gamma_I(D) +2.\]
\end{thm}
\begin{proof}
Since $\gamma_I(D) \geq 3$, it follows from Theorem \ref{mainthm2-5} that $\Delta^+(D) \leq n-2$.
Let $u$ be a vertex with $deg_D^+(u) =\Delta^+(D)$ and let $R:= \{uv \mid v \in V(D) \setminus N^+[u] \}$.
Then $(V(D) \setminus \{u\}, \emptyset, \{u\})$ is an IDF of $D +R$.
Thus,
\[r_I(D) \leq n - \Delta^+(D) - 1.\]
There exist $r_I(D) - 1$ vertices $v_1, v_2, \dotsc, v_{r_I(D)-1}$ in $V(D) \setminus N^+[u]$.

Let $D'$ be a digraph obtained from $D$ by adding $r_I(D) - 1$ arcs $uv_i$.
Then, by the definition of $r_I(D)$ and Observation \ref{easy1},
\[\gamma_I(D) = \gamma_I(D') \leq n - \Delta^+(D') + 1.\]
Since $\Delta^+(D') = \Delta^+(D) + r_I(D) -1$, we have $r_I(D) \leq n - \Delta^+(D) - \gamma_I(D) +2$.
\end{proof}

\begin{thm}
If $D$ is a digraph such that $\gamma_I(D) =3$ and $\gamma(D) =2$,
then $r(D) \leq r_I(D) + 1$.
\end{thm}
\begin{proof}
Let $R$ be a $r_I(D)$-set. Then $\gamma_I(D +R) =2$.
If $r \in R$, then clearly $r_I(D + (R\setminus \{r \})) = 1$.

By Theorem \ref{m-thm1},
there exist a $\gamma_I(D + (R\setminus \{r \}))$-function $f=(V_0, V_1, V_2)$ of $D + (R\setminus \{r \})$ and a vertex $v \in V_1$ satisfying one of the following conditions:
\begin{enumerate}
\item $f(N^-(v))=1$ and $f(N^-(x) \setminus \{v\}) \geq 2$ for each $x \in N^+(v) \cap V_0$.
\item $f(N^-(v))=0$, $f(N^-(x) \setminus \{v\}) \geq 2$ for each $x \in N^+(v)$, and $V_2 \neq \emptyset$.
\end{enumerate}

Suppose that (i) holds.
Since $\gamma_I(D + (R\setminus \{r \})) =3$, it follows from $f(N^-(v))=1$ that there  exists $u \in V_1$ such that $u \not\in N^-(v)$.
Let $w \in V_1 \cap N^-(v)$.
Since $f(N^-(x) \setminus \{v\}) \geq 2$ for each $x \in N^+(v) \cap V_0$,
we have $ux, wx \in A(D + (R\setminus \{r \})$ for each $x \in N^+(v) \cap V_0$.
Since $\gamma_I(D + (R\setminus \{r \})) =3$, we have $u, w \in N^-(x)$ for each $x \in V_0 \setminus N^+(v)$.
Thus, $\{u\}$ is a dominating set of $D + ((R\setminus \{r \}) \cup \{uv, uw\})$.
This implies that $r(D) \leq r_I(D) +1$.

Suppose that (ii) holds.
Let $V_2 = \{u \}$.
Then we have $V_0 \subseteq N^+(u)$.
Thus, $\{u\}$ is a dominating set of $D + ((R\setminus \{r \}) \cup \{uv\})$.
This implies that $r(D) \leq r_I(D)$.
\end{proof}

\subsection{The Italian reinforcement numbers of compositions of digraphs}\label{sec:main3-3}

For two digraphs $G$ and $H$, two kinds of joins $G \rightarrow H$ and $G \leftrightarrow H$ were defined in \cite{HWX}.
The digraph $G \rightarrow H$ consists of $G$ and $H$ with extra arcs from each vertex of $G$ to every vertex of $H$.
The digraph $G \leftrightarrow H$ can be obtained from $G \rightarrow H$ by adding arcs from each vertex of $H$ to every vertex of $G$.

\begin{thm}\label{com1}
Let $G$ and $H$ be two digraphs such that $\Delta^+(G) \geq 1$ and $\Delta^+(H) \geq 1$.
Then
\begin{enumerate}
\item $\gamma_I(G \rightarrow H) = \gamma_I(G)$,
\item $r_I(G \rightarrow H) = r_I(G)$,
\end{enumerate}
\end{thm}
\begin{proof}
(i) Let $f$ be a $\gamma_I(G)$-function.
Then it follows from the definition of IDF that $f$ is extended to an IDF of $G \rightarrow H$ by assigning $0$ to every vertex of $H$.
Thus, $\gamma_I(G \rightarrow H) \leq \gamma_I(G)$.
On the other hand, if $g=(V_0, V_1, V_2)$ is a $\gamma_I(G \rightarrow H)$-function,
then clearly $g|_G:= (V_0 \cap V(G), V_1 \cap V(G), V_2 \cap V(G))$ is an IDF of $G$.
Thus, $\gamma_I(G) \leq \gamma_I(G \rightarrow H)$.

\vskip5pt
(ii) If $\gamma_I(G) =2$, then it follows from (i) that $\gamma_I(G \rightarrow H) =2$. So, $r_I(G \rightarrow H) = r_I(G)$.
From now on, we assume $\gamma_I(G) \geq 3$.
Let $R$ be a $r_I(G)$-set.
Then
\[\gamma_I((G \rightarrow H) +R) = \gamma_I((G+R) \rightarrow H) = \gamma_I(G +R) < \gamma_I(G) = \gamma_I(G \rightarrow H).\]
Thus, $r_I(G \rightarrow H) \leq r_I(G)$.

Now we claim that $r_I(G) \leq r_I(G \rightarrow H)$.
Let $R_1$ be a $r_I(G \rightarrow H)$-set.
Suppose that $R_2$ is a subset of $R_1$ such that two ends of arcs in $R_2$ lie in $V(G)$.
Let $f = (V_0^f, V_1^f, V_2^f)$ be a $\gamma_I((G \rightarrow H) + R_1)$-function, and
let $g = f|_G$.
We divide our consideration into the following two cases.
\vskip5pt
\textbf{Case 1.} $g$ is an IDF of $G + R_2$.

Then we have
\begin{eqnarray*}
\gamma_I((G \rightarrow H) + R_1) &=& \omega(f) \\
                                  &\geq& \omega(g) \\
                                  &\geq& \gamma_I(G +R_2) \\
                                  &=&  \gamma_I((G +R_2) \rightarrow H) \\
                                  &=& \gamma_I((G \rightarrow H) +R_2) \\
                                  &\geq&  \gamma_I((G \rightarrow H) +R_1).
\end{eqnarray*}
Since $R_2 \subseteq R_1$ and $R_1$ is a $r_I(G \rightarrow H)$-set, we have $R_1 = R_2$.
So, $\gamma_I(G +R_2) \leq \omega(g) = \gamma_I((G \rightarrow H) +R_2) < \gamma_I(G \rightarrow H) = \gamma_I(G)$.
Thus, $r_I(G) \leq |R_2| = |R_1| = r_I(G \rightarrow H)$.

\vskip5pt
\textbf{Case 2.} $g$ is not an IDF of $G + R_2$.

Then some vertex $u \in V_0^f \cap V(G)$ has an in-neighbor $w \in V(H)$ such that $wu \in R_1$.
Fix $v \in V(G)$,
and let $R_3 =\{ vu \mid u \in N \}$,
where $N = \{ u \in V_0^f \cap V(G) \mid$ $u$ does not dominated by the vertices of $G ~\text{under}~ f\}$.
Then clearly $|R_2 \cup R_3| \leq |R_1|$.
It is easy to see that the function $h : V(G) \rightarrow \{0, 1, 2\}$ defined by
$h(v) = max\{ f(v), max\{ f(N^-(u) \cap V(H))  \mid  u \in N \}\}$ and $h(x)=f(x)$ otherwise,
is an IDF of $G +(R_2 \cup R_3)$ with weight at most $\omega(f)$.
Now we have
\begin{eqnarray*}
\gamma_I(G +(R_2 \cup R_3)) &\leq& \omega(h) \\
                                  &\leq& \omega(f) \\
                                  &=& \gamma_I((G \rightarrow H) +R_1) \\
                                  &<& \gamma_I(G \rightarrow H) \\
                                  &=&  \gamma_I(G).
\end{eqnarray*}
Thus, $r_I(G) \leq |R_2 \cup R_3| \leq |R_1| = r_I(G \rightarrow H)$.
\end{proof}

\vskip5pt
The corona $G \overrightarrow{\circ} H$ of two digraphs $G$ and $H$ is formed from one copy of $G$ and $n(G)$ copies of $H$ by joining $v_i$
to every vertex of $H_i$, where $v_i$ is the $i$th vertex of $G$ and $H_i$ is the $i$th copy of $H$.

\begin{thm}\label{com2}
Let $G$ and $H$ be two digraphs with $n(H) \geq 2$.
Then
\begin{enumerate}
\item $\gamma_I(G \overrightarrow{\circ} H) = 2n(G)$,
\item
\begin{equation*}\label{product}
r_I(G \overrightarrow{\circ} H) = \left\{
                      \begin{array}{ll}
                      0 & \hbox{if $n(G)=1$;} \\
                      n(H) & \hbox{if $G$ is the empty digraph and $n(G) \geq 2$;} \\
                      n(H) -1 & \hbox{otherwise.}
                      \end{array}
                     \right.
\end{equation*}
\end{enumerate}
\end{thm}
\begin{proof}
(i) If $n(G) =1$, then clearly $\gamma_I(G \overrightarrow{\circ} H) = 2$.
Assume that $n(G) \geq 2$.
It is easy to see that $(V(G \overrightarrow{\circ} H) \setminus V(G), \emptyset, V(G))$ is an IDF of $G \overrightarrow{\circ} H$.
So, $\gamma_I(G \overrightarrow{\circ} H) \leq 2n(G)$.

Let $f$ be a $\gamma_I(G \overrightarrow{\circ} H)$-function.
To dominate the vertices of $H_i$,we must have $\sum_{x \in V(H_i) \cup \{v_i\}} f(x) \geq 2$.
Since a single vertex of $G$ does not dominate vertices in different copies of $H$,
we have $\gamma_I(G \overrightarrow{\circ} H) \geq 2n(G)$.

\vskip5pt
(ii)
If $n(G) =1$, then clearly $r_I(G \overrightarrow{\circ} H) = 0$.
Assume that $n(G) \geq 2$.
We divide our consideration into the following two cases.

\vskip5pt
\textbf{Case 1.} $A(G) = \emptyset$.

Let $R = \{ v_1u \mid u \in V(H_{n(G)}) \}$.
Then it is easy to see that
\[(\bigcup_{i=1}^{n(G)} V(H_i), \{v_{n(G)} \}, V(G) \setminus \{v_{n(G)} \})\]
is an IDF of $(G \overrightarrow{\circ} H)+ R$ with weight $2n(G) -1$.
Thus, $r_I(G \overrightarrow{\circ} H) \leq n(H)$.

Let $F$ be a $r_I(G \overrightarrow{\circ} H)$-set.
By Lemma \ref{lem:b1}(ii), $\gamma_I((G \overrightarrow{\circ} H) +F) = \gamma_I(G \overrightarrow{\circ} H) -1$.
Let $U_i = \{v_i \} \cup V(H_i)$ for $1 \leq i \leq n(G)$, and let $f=(V_0, V_1, V_2)$ be a $r_I((G \overrightarrow{\circ} H) +F)$-function.
Then $\sum_{x \in U_i} f(x) \leq 1$ for some $i$, say $i=n(G)$.
To dominate the vertices in $U_{n(G)}$,
$F$ must contain at least $n(H)$ arcs which go from some vertices in $(V_1 \cup V_2) \cap (\bigcup_{i=1}^{n(G)-1}U_i)$
to vertices in $U_{n(G)}$.
Thus, $|F| \geq n(H)$ and so $r_I(G \overrightarrow{\circ} H) \geq n(H)$.

\vskip5pt
\textbf{Case 2.} $A(G) \neq \emptyset$.

Without loss of generality, we assume that $v_1v_{n(G)} \in A(G)$.
Let $V(H_{n(G)}) = \{ w_1, \dotsc, w_{n(H)} \}$, and let $R = \{ v_1w_j \mid w_j \in V(H_{n(G)}) \setminus \{ w_1\} \}$.
Then
\[(\bigcup_{i=1}^{n(G)} V(H_i) \cup \{v_{n(G)} \}, \{w_1\}, \{v_1, \dotsc, v_{n(G)-1}\} )\]
is an IDF of $(G \overrightarrow{\circ} H)+ R$ with weight $2n(G) -1$.
Thus, $r_I(G \overrightarrow{\circ} H) \leq n(H) -1$.
By using the same argument given in Case 1,
one can show that $r_I(G \overrightarrow{\circ} H) \geq n(H) -1$.
\end{proof}

\bibstyle{plain}

\end{document}